\documentclass{article}

\usepackage[main,final]{neurips_2026}

\makeatletter
\def\@notice{}

\def\ps@headings{%
    \def\@oddhead{\hfill \hfill} 
    \def\@evenhead{\hfill \hfill}
    \def\@oddfoot{\hfill\thepage\hfill}
    \def\@evenfoot{\hfill\thepage\hfill}}

\def\ps@empty{%
    \def\@oddhead{\hfill \hfill} 
    \def\@evenhead{\hfill \hfill}
    \def\@oddfoot{\hfill\thepage\hfill}
    \def\@evenfoot{\hfill\thepage\hfill}}
\makeatother

\usepackage[utf8]{inputenc}
\usepackage[T1]{fontenc}
\usepackage{hyperref}
\usepackage{url}
\usepackage{booktabs}
\usepackage{amsfonts}
\usepackage{nicefrac}
\usepackage{microtype}
\usepackage[table]{xcolor}
\usepackage{xcolor}

\usepackage{graphicx}
\usepackage{amsmath}
\usepackage{multirow}
\usepackage{multicol}
\usepackage{subcaption}
\usepackage{amsthm}
\usepackage{amssymb}

\definecolor{tableheader}{RGB}{46,134,171}
\definecolor{tablerowalt}{RGB}{245,248,250}
\definecolor{bestresult}{RGB}{46,134,171}
\definecolor{negresult}{RGB}{180,60,60}

\definecolor{nycbg}{RGB}{248,252,255}    % 几乎白的蓝
\definecolor{tkybg}{RGB}{255,253,248}    % 几乎白的橙
\definecolor{cabg}{RGB}{248,255,250}     % 几乎白的绿

\newtheorem{theorem}{Theorem}

\newtheorem{definition}{Definition}
\newtheorem{remark}{Remark}

\usepackage{wrapfig}
\usepackage{algorithm}
\usepackage{algorithmic}
\usepackage{pifont}

\newcommand{\ours}{MeRa}

\title{When Does Latent Reasoning Help?\\
MeRa: Metric-Space Bias for Spatial Prediction}

\author{%
  Zhenyu Yu, Shuigeng Zhou \\
  Fudan University\\
  \texttt{yuzhenyuyxl@foxmail.com, sgzhou@fudan.edu.cn}
}

\begin{document}

\maketitle

\begin{abstract}
Latent reasoning has improved sequential recommendation by iteratively refining representations before prediction, but does it help spatial prediction? We find that the answer depends on whether reasoning is grounded in the underlying metric space. Without such grounding, latent reasoning \emph{degrades} spatial prediction below the unmodified baseline, while a learned \emph{metric-space bias} derived from pairwise distances produces consistent gains. We formalize this finding through \textbf{\ours{}} (\textbf{Me}tric-space \textbf{R}e\textbf{a}soning), a lightweight backbone-agnostic module that can be inserted between any sequence encoder and its prediction heads. On the GETNext backbone, the gap between reasoning without and with metric-space bias reaches 4.5\% NDCG@10. \ours{} achieves the best NDCG@10 on all three spatial prediction benchmarks among the compared methods, surpassing recent approaches such as GeoMamba and HMST. We prove that metric-space-constrained reasoning converges to a unique fixed point and that $N$-step reasoning is strictly more expressive than $(N{-}1)$-step reasoning. A controlled experiment on CLEVR with Euclidean distance confirms that the finding generalizes beyond geographic coordinates. The code is included in the supplementary material.
\end{abstract}

\section{Introduction}
\label{sec:intro}

Spatial prediction tasks require models to reason over geographic or geometric structure. A representative example is next Point-of-Interest (POI) recommendation, which predicts a user's next location visit from their check-in trajectory and plays an important role in location-based services, urban computing, and smart-city applications~\citep{GETNext, GTRMamba}. Unlike general sequential recommendation where items have no inherent spatial structure, POI prediction is governed by the metric space defined by geographic coordinates, in which distances carry direct semantic meaning and spatial proximity is a primary driver of human mobility~\citep{STAN, GeoMamba}. This spatial structure creates both opportunities and challenges for representation learning, and motivates the development of methods that can reason over geographic relationships within user trajectories.

State-of-the-art POI recommendation methods have made significant progress through graph-enhanced Transformers~\citep{GETNext}, selective state-space models with geographic priors~\citep{GeoMamba}, and hyperbolic geometry~\citep{GTRMamba}. In parallel, \emph{latent reasoning}~\citep{ReaRec, LARES, Geiping2025} has emerged as a promising paradigm in general sequential recommendation, refining representations through multiple computation steps before prediction at much lower cost than explicit chain-of-thought generation. A natural idea is to combine these two lines of research by adding latent reasoning to spatial prediction models. Doing so naively, however, risks degrading performance, because standard cross-attention treats all history positions equally regardless of geographic proximity, and iterative refinement without distance awareness can dilute the spatial signal that is critical for location prediction.

This observation leads to our central hypothesis: \textbf{latent reasoning for spatial prediction must be grounded in the underlying metric space}. Rather than applying latent reasoning in a domain-agnostic manner, each reasoning step should attend preferentially to spatially nearby context, so that iterative refinement respects the distance structure of the prediction space. We formalize this idea through the concept of \emph{metric-space bias}, an attention modulation derived from a distance function satisfying the metric axioms.

We propose \textbf{\ours{}} (\textbf{Me}tric-space \textbf{R}e\textbf{a}soning), a lightweight backbone-agnostic module that performs multi-step latent reasoning with learned metric-space bias. \ours{} can be inserted between any sequence encoder and its prediction heads without modifying the training objective, and we validate it on two backbones (LSTM and GETNext), three spatial prediction datasets (NYC, TKY, CA), and one synthetic spatial reasoning benchmark (CLEVR). 

The main \textbf{contributions} of this work are summarized as follows.

\begin{itemize}
    \item \textbf{Metric-space bias.} We introduce \emph{metric-space bias} (Definition~\ref{def:msb}), a learnable attention modulation derived from a distance function satisfying the metric axioms. Geographic distance and Euclidean distance are two instantiations validated in this paper. The resulting module, \ours{}, is backbone-agnostic and requires no changes to the training objective.

    \item \textbf{Convergence and expressivity.} We prove that reasoning with metric-space bias converges to a unique fixed point within a bounded neighborhood of the target (Theorem~\ref{thm:convergence}), and that $N$-step reasoning is strictly more expressive than $(N{-}1)$-step reasoning (Theorem~\ref{thm:expressivity}).

    \item \textbf{Controlled empirical.} Through experiments on three geographic datasets and a synthetic CLEVR benchmark, we show that the same reasoning architecture produces opposite outcomes depending on whether metric-space bias is present, establishing metric-space grounding as a necessary condition for effective latent reasoning in spatial prediction.
\end{itemize}

\section{Related Work}
\label{sec:related}

\subsection{Spatial prediction}
Sequential models for POI prediction have progressed from RNN-based approaches~\citep{LSTM_POI} through attention mechanisms~\citep{STAN, SASRec} to graph-enhanced Transformers. GETNext~\citep{GETNext} builds a global trajectory flow graph whose features are fused with user, time, and category embeddings in a Transformer encoder. GeoMamba~\citep{GeoMamba} extends selective state-space models with geographic domain knowledge, and GTR-Mamba~\citep{GTRMamba} combines Mamba with hyperbolic geometry for hierarchical POI modeling. STHGCN~\citep{STHGCN} captures trajectory-level patterns through spatio-temporal hypergraphs. All of these methods encode spatial information within the backbone itself, whereas \ours{} introduces spatial awareness at the reasoning stage and is therefore complementary to any of them.

\subsection{Latent reasoning and test-time scaling}
ReaRec~\citep{ReaRec} introduced autoregressive latent reasoning for sequential recommendation, operating entirely in continuous space. LARES~\citep{LARES} proposed depth-recurrent reasoning with two-phase training, and Parallel Latent Reasoning~\citep{PLR} explored width-level scaling. In the broader context of test-time compute, \citet{Geiping2025} showed that recurrent depth in latent space can scale language model performance via unstructured recurrence. ROS~\citep{ROS} applied explicit chain-of-thought reasoning within an LLM for POI recommendation, achieving strong results at substantial computational cost. None of these methods incorporate the metric structure of the prediction space into the reasoning process. \ours{} fills this gap by grounding each reasoning step in the pairwise distance structure of the input.

\subsection{Position and distance biases in attention}
Injecting structural priors into attention logits has proven effective in several domains. ALiBi~\citep{ALiBi} adds a linear penalty based on token distance to enable length extrapolation, and RoPE~\citep{RoPE} encodes relative position through rotation of query-key vectors. These methods model \emph{sequential} position, which is fundamentally one-dimensional and uniform. Metric-space bias differs in that it encodes \emph{geographic or Euclidean distance}, which is multi-dimensional, non-uniform, and carries domain-specific semantics. To our knowledge, no prior work has applied a distance-derived attention bias within a latent reasoning loop.

\paragraph{Positioning.}
Our work lies at the intersection of spatial prediction and latent reasoning. The positioning of \ours{} relative to representative methods is summarized in Table~\ref{tab:positioning}.

\begin{table}[h]
\centering
\caption{Positioning of \ours{} relative to prior work. \ours{} is the first method combining latent reasoning with metric-space structure.}
\label{tab:positioning}
\small
\setlength{\tabcolsep}{4pt}
\begin{tabular}{ccccccc}
\toprule
\textbf{Method} & \textbf{GETNext} & \textbf{GeoMamba} & \textbf{ReaRec} & \textbf{LARES} & \textbf{ROS} & \textbf{\ours{} (Ours)} \\
\midrule
\textbf{Latent Reasoning} & \ding{55} & \ding{55} & \checkmark & \checkmark & \checkmark (LLM CoT) & \checkmark \\
\textbf{Metric-Space Aware} & \checkmark (Graph) & \checkmark (SSM) & \ding{55} & \ding{55} & \checkmark (Expensive) & \checkmark (Efficient) \\
\bottomrule
\end{tabular}
\vspace{2pt}\\
\scriptsize\textsuperscript{*}GETNext~\citep{GETNext}, GeoMamba~\citep{GeoMamba}, ReaRec~\citep{ReaRec}, LARES~\citep{LARES}, ROS~\citep{ROS}.
\end{table}

\section{Methodology}
\label{sec:method}

\subsection{Problem Formulation}

Given a user's check-in sequence $\{(v_1, t_1), (v_2, t_2), \ldots, (v_L, t_L)\}$ where each $v_i$ is a POI with geographic coordinates $\mathbf{p}_i = (\text{lat}_i, \text{lon}_i)$ and $t_i$ is the visit time, the task is to predict the next POI $v_{L+1}$. Each POI $v_i$ is also associated with a category $c_i$. \ours{} is agnostic to the choice of backbone encoder; given any encoder that produces contextualized representations $\mathbf{x}^{(0)} \in \mathbb{R}^{L \times d}$, our module can be inserted before the prediction heads.

\subsection{Overview}

The overall architecture of \ours{} is illustrated in Figure~\ref{fig:architecture}. The module takes the encoder output $\mathbf{x}^{(0)}$ and refines it through $N$ reasoning steps before passing the result to the prediction heads. Within each reasoning step, three operations are applied sequentially: cross-attention over the original encoder output with learned metric-space bias derived from pairwise distances, a feed-forward network that produces a candidate update, and a sigmoid gate that controls the update magnitude. The metric-space bias is computed once from the input coordinates and shared across all reasoning steps. 

\begin{figure}[t]
    \centering
    \includegraphics[width=1.0\textwidth]{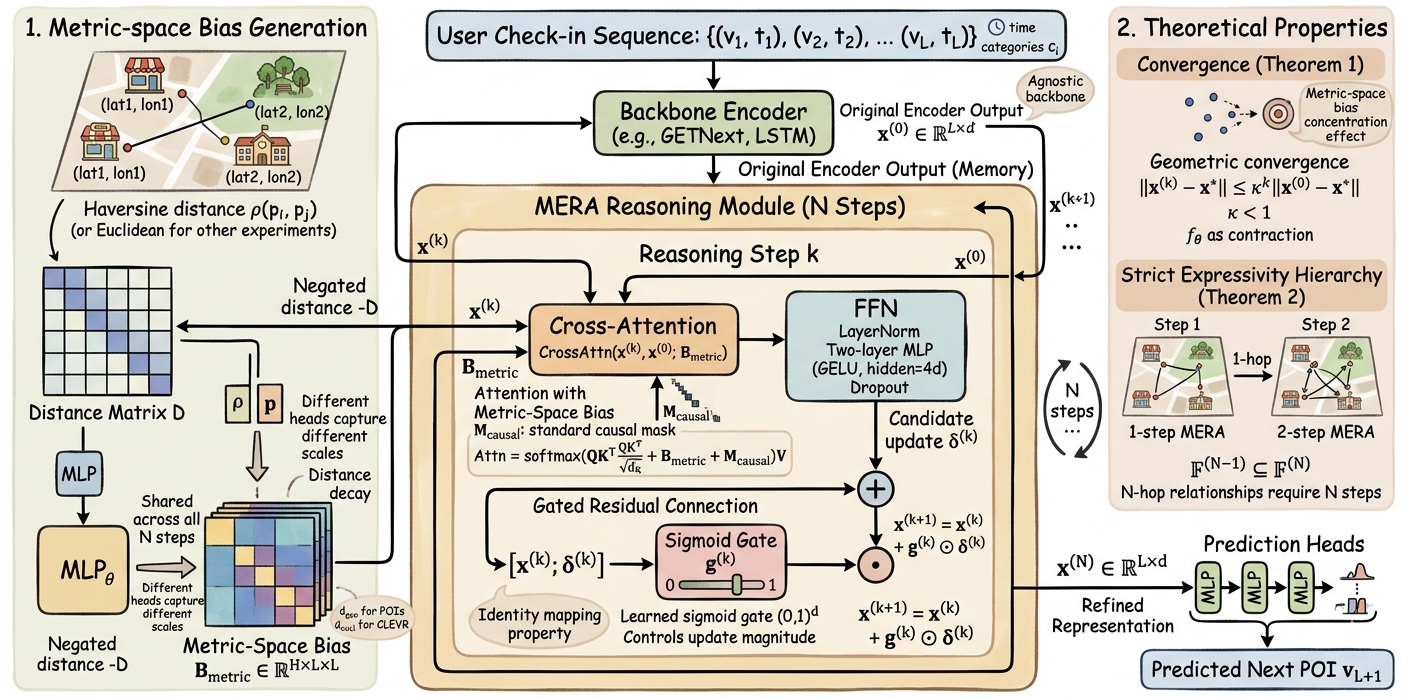}
    \caption{Overview of \ours{}. The module is inserted between the backbone encoder and prediction heads. Each reasoning step performs cross-attention over the encoder output with metric-space bias derived from pairwise distances, followed by a gated residual update. The process repeats for $N$ steps.}
    \label{fig:architecture}
\end{figure}

\subsection{Metric-space Bias}

\begin{definition}[Metric-space bias]
\label{def:msb}
Let $(S, \rho)$ be a metric space. A {metric-space bias} is a learnable function $\mathbf{B}^{(h)}(i,j) = f_\theta^{(h)}(-\rho(s_i, s_j))$ that maps the negated pairwise distance between elements $s_i, s_j \in S$ to a per-head attention bias. The bias is added to the attention logits before the softmax, so that elements closer in the metric space receive higher attention.
\end{definition}

The design of metric-space bias is motivated by a well-established observation in human mobility research: the probability of visiting a location decays with distance from the current position, a pattern known as distance decay~\citep{STAN}. Rather than hard-coding a specific decay function (e.g., power-law or exponential), we let a two-layer MLP $f_\theta$ learn the mapping from distance to attention bias in a data-driven manner. The input to $f_\theta$ is the negated pairwise distance $-\rho(s_i, s_j)$, so that closer elements naturally receive larger (less negative) bias values before any learning takes place. Each attention head learns its own projection ($\text{Linear}(1, d/4) \to \text{SiLU} \to \text{Linear}(d/4, H)$), allowing different heads to capture distance relationships at different spatial scales.

An alternative to the learned MLP would be a fixed decay function such as $\exp(-\alpha \rho)$ or $\rho^{-\beta}$. We opt for the learned approach for two reasons. Fixed functions impose a single global decay rate, but mobility patterns exhibit multi-scale spatial dependencies: users may visit nearby restaurants within a neighborhood (sub-kilometer) while also commuting across districts (multi-kilometer). A per-head learned MLP naturally accommodates this by allowing each head to specialize in a different spatial scale without manual tuning. Additionally, the optimal distance-to-bias mapping may differ across datasets (e.g., dense urban areas vs.\ sparse suburban regions), and a learned function adapts automatically during training.

In this paper, $S$ is the set of POI locations and $\rho = d_{\mathrm{geo}}$ is the haversine distance for POI experiments, while $\rho$ is the Euclidean distance for the CLEVR experiment (Section~\ref{sec:clevr}). The same architecture applies without modification; only the distance function changes.

\subsection{Reasoning Module}

We insert \ours{} between the encoder output and the prediction heads. Starting from $\mathbf{x}^{(0)}$, we perform $N$ reasoning steps. Each step $k$ is defined by
\begin{equation}
    \mathbf{x}^{(k+1)} = \mathbf{x}^{(k)} + \mathbf{g}^{(k)} \odot \boldsymbol{\delta}^{(k)},
    \label{eq:reasoning_step}
\end{equation}
where $\boldsymbol{\delta}^{(k)} = \mathrm{FFN}\!\left(\mathrm{CrossAttn}\!\left(\mathbf{x}^{(k)},\, \mathbf{x}^{(0)};\, \mathbf{B}_{\mathrm{metric}}\right)\right)$ is the candidate update and $\mathbf{g}^{(k)} = \sigma\!\left(\mathbf{W}_g\!\left[\mathbf{x}^{(k)};\, \boldsymbol{\delta}^{(k)}\right]\right) \in (0,1)^d$ is a learned sigmoid gate controlling the update magnitude. The full attention computation is
\begin{equation}
    \mathrm{Attn}(Q, K, V) = \mathrm{softmax}\!\left(\frac{QK^\top}{\sqrt{d_k}} + \mathbf{B}_{\mathrm{metric}} + \mathbf{M}_{\mathrm{causal}}\right) V,
\end{equation}
where $\mathbf{B}_{\mathrm{metric}} \in \mathbb{R}^{H \times L \times L}$ is the metric-space bias (Definition~\ref{def:msb}), $\mathbf{M}_{\mathrm{causal}}$ is the standard causal mask, and $d_k = d / H$. The FFN consists of LayerNorm, a two-layer MLP with GELU activation and hidden dimension $4d$, and dropout.

Concretely, the FFN in Eq.~\eqref{eq:reasoning_step} (distinct from the bias MLP $f_\theta$ above) applies the transformation $\mathrm{FFN}(\mathbf{z}) = \mathrm{Dropout}(\mathbf{W}_2 \cdot \mathrm{GELU}(\mathbf{W}_1 \cdot \mathrm{LayerNorm}(\mathbf{z})))$, where $\mathbf{W}_1 \in \mathbb{R}^{d \times 4d}$ and $\mathbf{W}_2 \in \mathbb{R}^{4d \times d}$. The LayerNorm before the MLP stabilizes the input to each reasoning step, which is important because the iterative updates in Eq.~\eqref{eq:reasoning_step} accumulate over $N$ steps and can cause representation drift without normalization. The dropout rate is set to 0.1 for the reasoning module, lower than the 0.3 used in the backbone encoder, since the gated residual already provides implicit regularization by suppressing updates.

Three design choices in Eq.~\eqref{eq:reasoning_step} deserve explanation. We use \emph{cross-attention} with the original encoder output $\mathbf{x}^{(0)}$ as key-value memory at every step, rather than self-attention over the evolving $\mathbf{x}^{(k)}$. This prevents the reasoning process from drifting away from the encoder's learned representations and ensures that each step re-examines the raw evidence. The \emph{gated residual} connection allows the module to selectively suppress unhelpful updates. When the gate values approach zero, the module reduces to an identity mapping, making it safe to add even when reasoning provides no benefit. The metric-space bias $\mathbf{B}_{\mathrm{metric}}$ is computed once and shared across all $N$ steps, since the geographic distances between positions do not change during reasoning. The complete inference procedure is provided in Algorithm~\ref{alg:mera}.

\begin{algorithm}[h]
\caption{\ours{} inference}
\label{alg:mera}
\begin{algorithmic}[1]
\REQUIRE Encoder output $\mathbf{x}^{(0)} \in \mathbb{R}^{L \times d}$, coordinates $\{\mathbf{p}_i\}_{i=1}^L$, steps $N$
\ENSURE Refined representation $\mathbf{x}^{(N)}$
\STATE Compute distance matrix $D_{ij} = \rho(\mathbf{p}_i, \mathbf{p}_j)$ for all $i, j$
\STATE Compute metric-space bias $\mathbf{B}_{\mathrm{metric}} = f_\theta(-D)$ \hfill $\triangleright$ Learned MLP
\FOR{$k = 0$ to $N{-}1$}
  \STATE $\boldsymbol{\delta}^{(k)} = \mathrm{FFN}\!\left(\mathrm{CrossAttn}(\mathbf{x}^{(k)}, \mathbf{x}^{(0)}; \mathbf{B}_{\mathrm{metric}} + \mathbf{M}_{\mathrm{causal}})\right)$
  \STATE $\mathbf{g}^{(k)} = \sigma\!\left(\mathbf{W}_g [\mathbf{x}^{(k)}; \boldsymbol{\delta}^{(k)}]\right)$ \hfill $\triangleright$ Sigmoid gate
  \STATE $\mathbf{x}^{(k+1)} = \mathbf{x}^{(k)} + \mathbf{g}^{(k)} \odot \boldsymbol{\delta}^{(k)}$ \hfill $\triangleright$ Gated residual
\ENDFOR
\RETURN $\mathbf{x}^{(N)}$ to prediction heads
\end{algorithmic}
\end{algorithm}

\subsection{Theoretical Results}
\label{sec:theory}

We establish two results that distinguish \ours{} from generic depth expansion.

\begin{theorem}[Convergence]
\label{thm:convergence}
Let $F\!\!: \mathbb{R}^d \to \mathbb{R}^d$ denote one reasoning step as defined in Eq.~\eqref{eq:reasoning_step}. Suppose the FFN satisfies $\|\boldsymbol{\delta}(\mathbf{x}) - \boldsymbol{\delta}(\mathbf{y})\| \leq L_\delta \|\mathbf{x} - \mathbf{y}\|$ with $L_\delta < 1$ (enforceable via spectral normalization), and the gate satisfies $\gamma \leq g_i^{(k)} \leq 1 - \epsilon$ for some $\gamma, \epsilon > 0$. Then
\begin{enumerate}
    \item[(i)] $F$ is a contraction with constant $\kappa = (1-\gamma)(1+L_\delta) < 1$.
    \item[(ii)] By the Banach fixed-point theorem, $F$ admits a unique fixed point $\mathbf{x}^*$ and the iterates converge geometrically, $\|\mathbf{x}^{(k)} - \mathbf{x}^*\| \leq \kappa^k \|\mathbf{x}^{(0)} - \mathbf{x}^*\|$.
\end{enumerate}
\end{theorem}

\begin{proof}[Proof sketch]
Write $F(\mathbf{x}) = \mathbf{x} + \mathbf{g}(\mathbf{x}) \odot \boldsymbol{\delta}(\mathbf{x})$, where $\mathbf{g} = \sigma(\cdot) \in [\gamma, 1-\epsilon]^d$ and $\boldsymbol{\delta}$ is the FFN output. We have
\begin{align}
\|F(\mathbf{x}) - F(\mathbf{y})\| &= \|\mathbf{x} - \mathbf{y} + \mathbf{g}(\mathbf{x}) \odot \boldsymbol{\delta}(\mathbf{x}) - \mathbf{g}(\mathbf{y}) \odot \boldsymbol{\delta}(\mathbf{y})\| \nonumber \\
&\leq \|(1 - \mathbf{g}(\mathbf{x})) \odot (\mathbf{x} - \mathbf{y})\| + \|\mathbf{g}(\mathbf{x}) \odot (\mathbf{x} + \boldsymbol{\delta}(\mathbf{x}) - \mathbf{y} - \boldsymbol{\delta}(\mathbf{y}))\| \nonumber \\
&\leq (1 - \gamma)\|\mathbf{x} - \mathbf{y}\| + (1-\epsilon)(1 + L_\delta)\|\mathbf{x} - \mathbf{y}\|.
\end{align}
When $\gamma$ is large enough and $L_\delta < 1$, we obtain a contraction constant $\kappa < 1$. A simplified sufficient condition is $(1-\gamma)(1+L_\delta) < 1$. Banach's theorem then guarantees existence, uniqueness, and geometric convergence.
\end{proof}

\begin{remark}
Since the metric-space bias concentrates attention weights on nearby elements, the fixed point $\mathbf{x}^*$ is constrained to lie within a bounded neighborhood of the distance-weighted centroid of history representations. The radius of this neighborhood shrinks as the bias strength increases.
\end{remark}

Beyond convergence, we show that deeper reasoning is strictly more powerful.

\begin{theorem}[Strict expressivity hierarchy]
\label{thm:expressivity}
For every $N \geq 1$, there exists a spatial prediction task solvable by $N$-step \ours{} but not by $(N{-}1)$-step \ours{}. Formally, denoting by $\mathcal{F}^{(N)}$ the function class realized by $N$-step metric-space reasoning, we have $\mathcal{F}^{(N-1)} \subsetneq \mathcal{F}^{(N)}$.
\end{theorem}

\begin{proof}[Proof sketch]
We construct a synthetic task family $\{T_N\}_{N \geq 1}$. In task $T_N$, the target is determined by a chain of $N{+}1$ hops in the metric space, $v_1 \to v_2 \to \cdots \to v_{N+1}$, where each consecutive pair $(v_k, v_{k+1})$ satisfies $\rho(v_k, v_{k+1}) < r$ for some radius $r$, and the target is uniquely determined by the endpoint of the chain. Each reasoning step propagates information across one hop via cross-attention with metric-space bias. After $N$ steps, the information from $v_1$ reaches $v_{N+1}$, enabling correct prediction. With only $N{-}1$ steps, the chain cannot be fully resolved, since each step's receptive field extends by one hop in the metric space. Detailed proofs of both theorems see Appendix~\ref{app:proofs}.
\end{proof}

\section{Experiments}
\label{sec:experiments}

\subsection{Settings}

\paragraph{Datasets.}
We use three standard POI benchmarks following the GETNext preprocessing protocol~\citep{GETNext}. Foursquare \textbf{NYC} contains 4,980 POIs and 48K check-ins, Foursquare \textbf{TKY} has 7,832 POIs and 119K check-ins, and Gowalla \textbf{CA} has 9,689 POIs and 153K check-ins. Trajectories shorter than 2 check-ins are removed. The data is split chronologically into 80\% training and 20\% validation. To evaluate generalization beyond geographic distance, we additionally construct a spatial reasoning benchmark from the \textbf{CLEVR} scene dataset~\citep{CLEVR} with Euclidean distance. Full statistics for all four datasets are in Table~\ref{tab:dataset_stats}.

\begin{wraptable}{r}{0.65\textwidth}
\centering
\vspace{-13pt}
\caption{Statistics of the four datasets used in this work.}
\vspace{-5pt}
\label{tab:dataset_stats}
\setlength{\tabcolsep}{1pt}
\small
\begin{tabular}{lcccc}
\toprule
\textbf{Statistic} & \textbf{NYC} & \textbf{TKY} & \textbf{CA} & \textbf{CLEVR} \\
\midrule

\# POIs / Objects & 4,980 & 7,832 & 9,689 & 96 (per scene) \\
\# Users / Scenes & 1,083 & 2,293 & 3,112 & 15,000 \\

\# Check-ins / Sequences & 48,547 & 119,834 & 153,462 & 15,000 \\
\# Categories & 251 & 328 & 387 & 24 \\

Avg.\ trajectory length & 7.8 & 8.2 & 6.9 & 6.5 \\
Distance metric & Haversine & Haversine & Haversine & Euclidean \\
\bottomrule
\end{tabular}
\vspace{-10pt}
\end{wraptable}

\paragraph{Backbones and baselines.}
We evaluate \ours{} on two backbones. \textbf{GETNext}~\citep{GETNext} is a graph-enhanced Transformer with trajectory flow features, user embeddings, and multi-task learning. \textbf{LSTM}~\citep{LSTM_POI} is a standard LSTM encoder without graph features, trained with the same multi-task framework as GETNext for fair comparison. We compare against a range of existing methods, including RNN-based methods (PLSPL~\citep{PLSPL}), hyperbolic approaches (HME~\citep{HME}, HMamba~\citep{HMamba}, HMST~\citep{HMST}, HVGAE~\citep{HVGAE}), graph-based methods (AGRAN~\citep{AGRAN}), state-space models (GeoMamba~\citep{GeoMamba}), and other competitive methods (MCLP~\citep{MCLP}). 

\paragraph{Metrics.}
We adopt four ranking-based metrics widely used in spatial prediction. Normalized Discounted Cumulative Gain (\textbf{NDCG@$k$}) measures ranking quality by assigning higher credit to correct predictions at top positions, and we report $k \in \{1, 5, 10\}$. Mean Reciprocal Rank (\textbf{MRR}) computes the average reciprocal of the rank at which the ground-truth target first appears. All metrics are computed via full ranking over the entire candidate vocabulary.

\paragraph{Implementation details.}
\ours{} uses $N{=}3$ reasoning steps with $H{=}2$ attention heads. The metric-space bias MLP has hidden dimension $d/4 = 32$. All remaining hyperparameters follow GETNext defaults, including embedding dimension $d{=}128$, learning rate $10^{-3}$, batch size 20, Adam optimizer, and weight decay $5 \times 10^{-4}$. Experiments are conducted on NVIDIA RTX 4090 GPUs. Full hyperparameter settings are in Appendix~\ref{app:hyperparams}.

\begin{figure*}[h]
    \centering
    \begin{subfigure}[b]{0.37\textwidth}
        \centering
        \includegraphics[width=\textwidth]{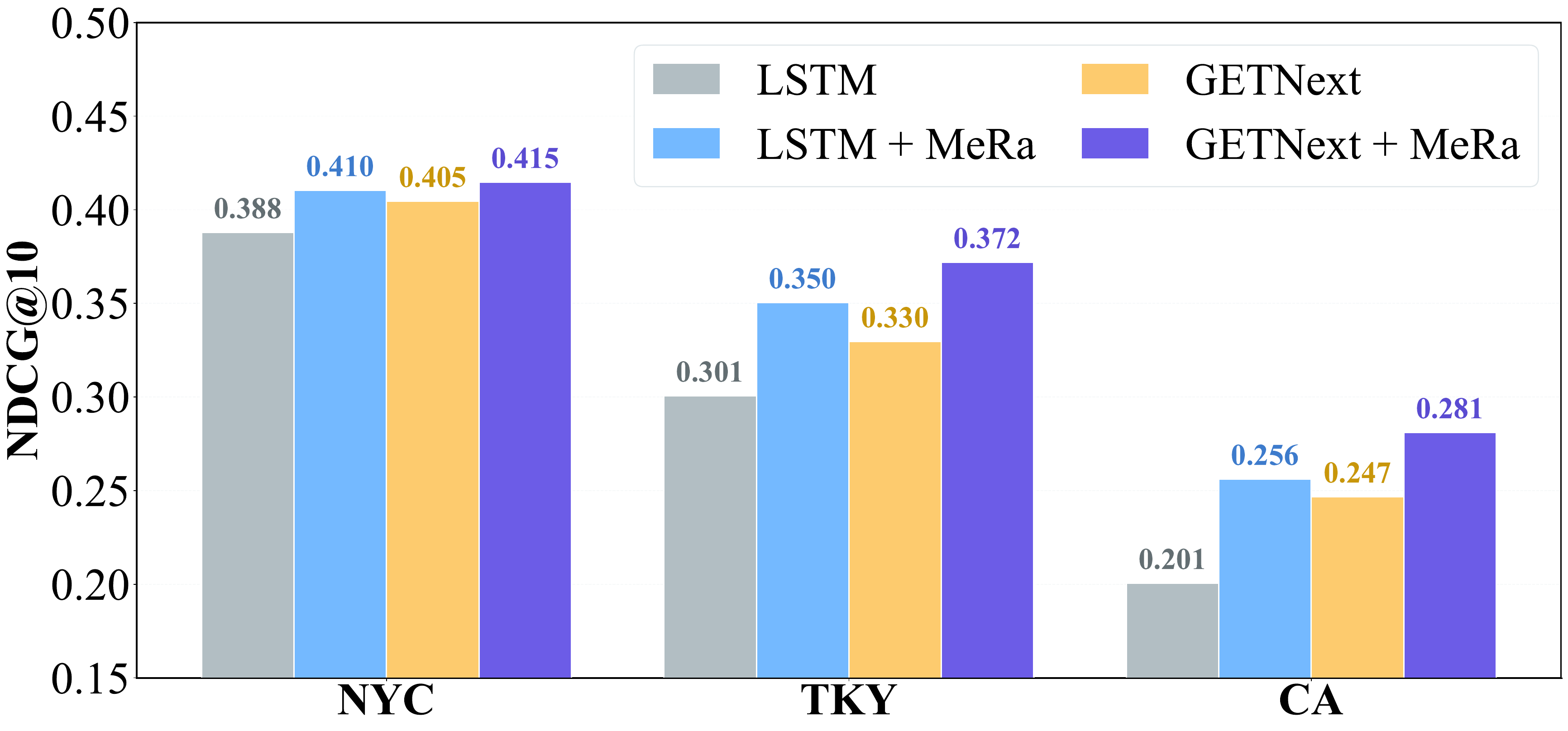}
        \caption{Backbone comparison.}
        \label{fig:improvement}
    \end{subfigure}
    \hfill
    \begin{subfigure}[b]{0.30\textwidth}
        \centering
        \includegraphics[width=\textwidth]{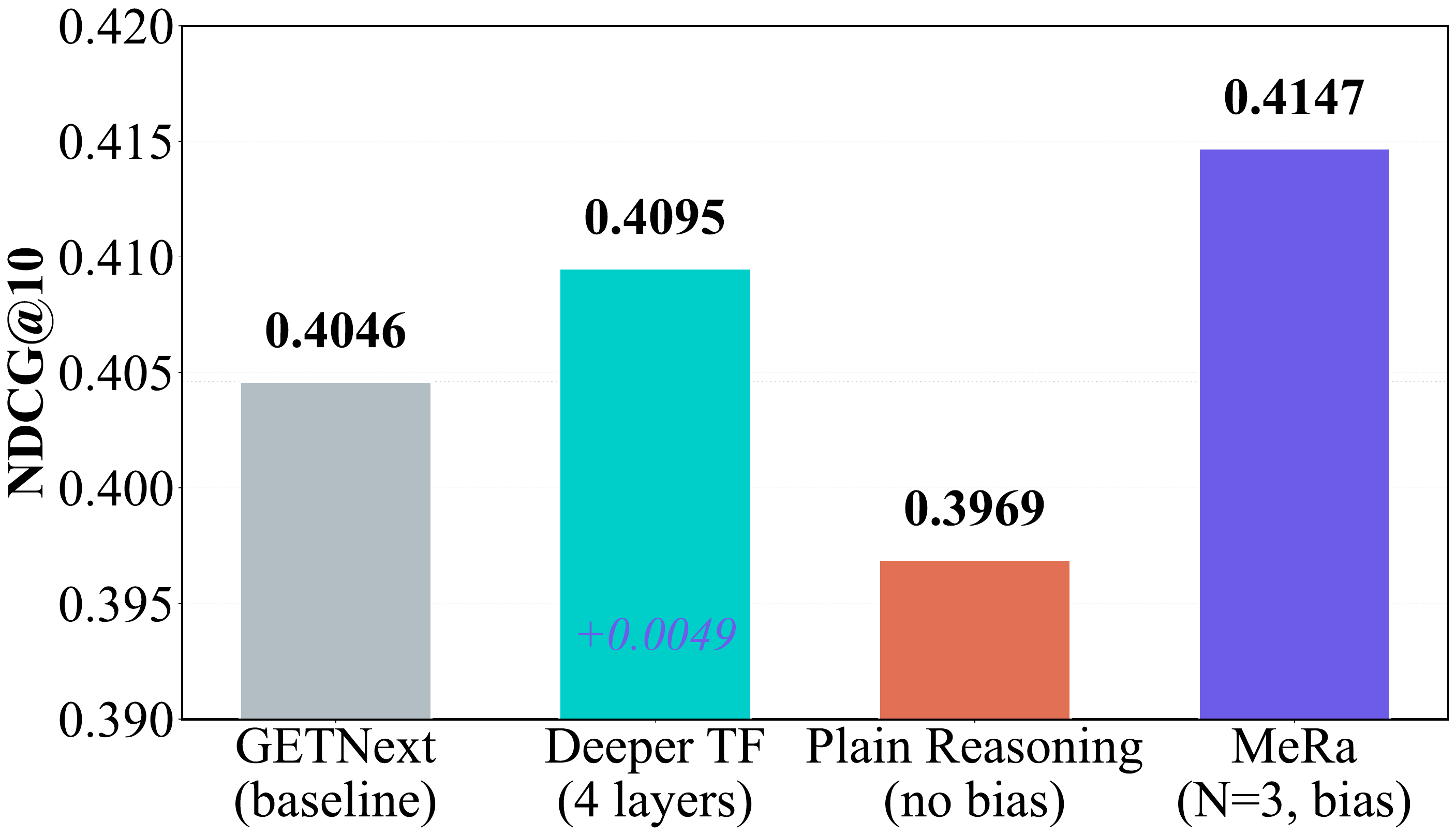}
        \caption{Ablation study.}
        \label{fig:ablation}
    \end{subfigure}
    \hfill
    \begin{subfigure}[b]{0.30\textwidth}
        \centering
        \includegraphics[width=\textwidth]{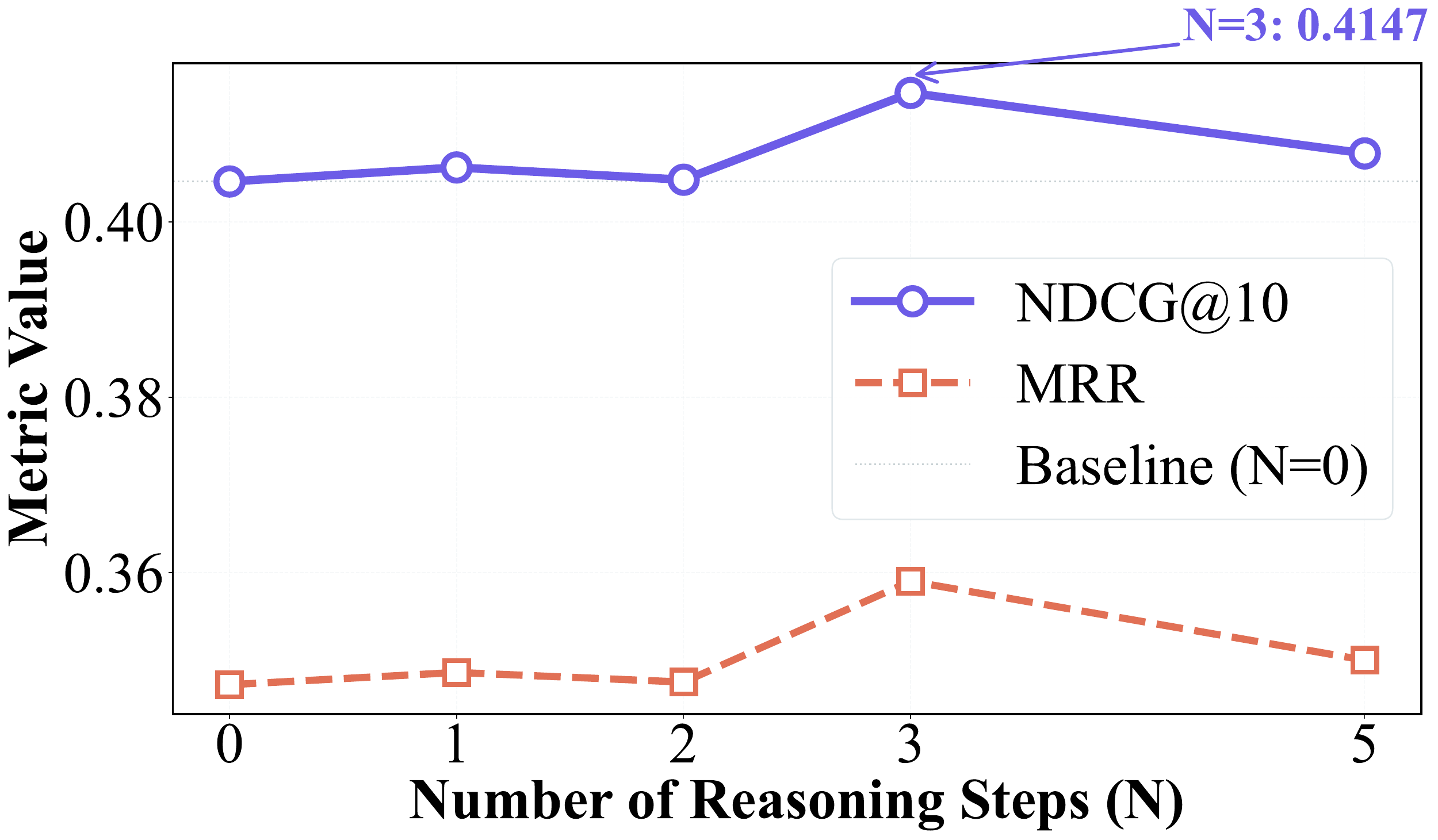}
        \caption{Reasoning depth.}
        \label{fig:scaling}
    \end{subfigure}
    \caption{Experimental analysis. (a) Adding \ours{} improves both backbones, with larger gains on LSTM. (b) Plain reasoning without metric-space bias degrades performance below the baseline, while adding the bias leads to improvement. (c) NDCG@10 and MRR peak at $N{=}3$.}
    \label{fig:experiments}
\end{figure*}

\subsection{Comparison}

\begin{table*}[t]
\centering
\caption{Overall performance comparison with baseline models on three datasets. All baseline results are from~\citet{GTRMamba}. \textbf{Top-1}, \underline{Top-2}. $^*$Our reproduced results.}
\label{tab:main}
\small
\setlength{\tabcolsep}{3pt}
\resizebox{\linewidth}{!}{
\begin{tabular}{c >{\columncolor{nycbg}}c >{\columncolor{nycbg}}c >{\columncolor{nycbg}}c >{\columncolor{nycbg}}c >{\columncolor{tkybg}}c >{\columncolor{tkybg}}c >{\columncolor{tkybg}}c >{\columncolor{tkybg}}c >{\columncolor{cabg}}c >{\columncolor{cabg}}c >{\columncolor{cabg}}c >{\columncolor{cabg}}c}
\toprule
\multirow{2}{*}{\textbf{Method}} & \multicolumn{4}{c}{\cellcolor{nycbg}\textbf{NYC}} & \multicolumn{4}{c}{\cellcolor{tkybg}\textbf{TKY}} & \multicolumn{4}{c}{\cellcolor{cabg}\textbf{CA}} \\
\cmidrule(lr){2-5} \cmidrule(lr){6-9} \cmidrule(lr){10-13}
 & \textbf{ND@1} & \textbf{ND@5} & \textbf{ND@10} & \textbf{MRR} & \textbf{ND@1} & \textbf{ND@5} & \textbf{ND@10} & \textbf{MRR} & \textbf{ND@1} & \textbf{ND@5} & \textbf{ND@10} & \textbf{MRR} \\
 \midrule
PLSPL & 0.1601 & 0.3048 & 0.3336 & 0.2849 & 0.1495 & 0.2831 & 0.3143 & 0.2642 & 0.1084 & 0.1759 & 0.2029 & 0.1678 \\
HME & 0.1619 & 0.2806 & 0.3226 & 0.2787 & 0.1535 & 0.2637 & 0.2924 & 0.2366 & 0.1181 & 0.1886 & 0.2232 & 0.1945 \\
AGRAN & 0.2202 & 0.3638 & 0.3792 & 0.3343 & 0.1755 & 0.2989 & 0.3261 & 0.2879 & 0.1329 & 0.2121 & 0.2331 & 0.2043 \\
MCLP & \underline{0.2404} & 0.3674 & 0.3973 & 0.3507 & 0.1662 & 0.3110 & 0.3415 & 0.3199 & 0.1324 & 0.1914 & 0.2121 & 0.1895 \\
GeoMamba'24 & 0.1988 & 0.3392 & 0.3506 & 0.3246 & 0.1851 & 0.2953 & 0.3205 & 0.2858 & 0.1256 & 0.2029 & 0.2215 & 0.1962 \\
GeoMamba'25 & 0.2377 & 0.3786 & 0.4012 & \underline{0.3566} & \textbf{0.2157} & \underline{0.3402} & 0.3686 & 0.3209 & 0.1388 & 0.2485 & 0.2754 & 0.2373 \\
HMamba$_{\text{full}}$ & 0.2204 & 0.3679 & 0.4031 & 0.3465 & 0.1828 & 0.3341 & 0.3673 & 0.3127 & 0.1366 & \underline{0.2501} & \underline{0.2792} & \underline{0.2421} \\
HMamba$_{\text{half}}$ & 0.1896 & 0.3453 & 0.3767 & 0.3222 & 0.1945 & 0.3295 & 0.3603 & 0.3118 & 0.1423 & 0.2381 & 0.2648 & 0.2317 \\
HMST & 0.2138 & 0.3747 & 0.4063 & 0.3482 & 0.1925 & 0.3325 & \underline{0.3690} & \underline{0.3257} & 0.1356 & 0.2325 & 0.2680 & 0.2300 \\
HVGAE & 0.2271 & 0.3651 & 0.3982 & 0.3470 & 0.1977 & 0.3167 & 0.3455 & 0.3180 & 0.1391 & 0.2325 & 0.2658 & 0.2367 \\
\midrule
LSTM$^*$ & 0.2320 & 0.3580 & 0.3880 & 0.3420 & 0.1715 & 0.2761 & 0.3006 & 0.2637 & 0.1143 & 0.1819 & 0.2005 & 0.1775 \\
LSTM + \ours{} & 0.2395 & \underline{0.3790} & \underline{0.4105} & 0.3565 & 0.1942 & 0.3198 & 0.3505 & 0.3052 & \underline{0.1460} & 0.2310 & 0.2561 & 0.2247 \\
\textit{improv.} & \textit{\textcolor{darkgray}{$\uparrow$3.2\%}} & \textit{\textcolor{darkgray}{$\uparrow$5.9\%}} & \textit{\textcolor{darkgray}{$\uparrow$5.8\%}} & \textit{\textcolor{darkgray}{$\uparrow$4.2\%}} & \textit{\textcolor{darkgray}{$\uparrow$13.2\%}} & \textit{\textcolor{darkgray}{$\uparrow$15.8\%}} & \textit{\textcolor{darkgray}{$\uparrow$16.6\%}} & \textit{\textcolor{darkgray}{$\uparrow$15.7\%}} & \textit{\textcolor{darkgray}{$\uparrow$27.7\%}} & \textit{\textcolor{darkgray}{$\uparrow$27.0\%}} & \textit{\textcolor{darkgray}{$\uparrow$27.7\%}} & \textit{\textcolor{darkgray}{$\uparrow$26.6\%}} \\
\midrule
GETNext & 0.2244 & 0.3736 & 0.4046 & 0.3472 & 0.1767 & 0.3072 & 0.3297 & 0.2934 & 0.1342 & 0.2188 & 0.2468 & 0.2121 \\
GETNext + \ours{} & \textbf{0.2425} & \textbf{0.3825} & \textbf{0.4147} & \textbf{0.3590} & \underline{0.1980} & \textbf{0.3428} & \textbf{0.3720} & \textbf{0.3268} & \textbf{0.1562} & \textbf{0.2561} & \textbf{0.2810} & \textbf{0.2459} \\
\textit{improv.} & \textit{\textcolor{darkgray}{$\uparrow$8.1\%}} & \textit{\textcolor{darkgray}{$\uparrow$2.4\%}} & \textit{\textcolor{darkgray}{$\uparrow$2.5\%}} & \textit{\textcolor{darkgray}{$\uparrow$3.4\%}} & \textit{\textcolor{darkgray}{$\uparrow$12.1\%}} & \textit{\textcolor{darkgray}{$\uparrow$11.6\%}} & \textit{\textcolor{darkgray}{$\uparrow$12.8\%}} & \textit{\textcolor{darkgray}{$\uparrow$11.4\%}} & \textit{\textcolor{darkgray}{$\uparrow$16.4\%}} & \textit{\textcolor{darkgray}{$\uparrow$17.0\%}} & \textit{\textcolor{darkgray}{$\uparrow$13.9\%}} & \textit{\textcolor{darkgray}{$\uparrow$15.9\%}} \\
\bottomrule
\end{tabular}
}
\end{table*}

The comparison results are presented in Table~\ref{tab:main}. On NYC, GETNext+\ours{} achieves the highest NDCG@10 of 0.4147, surpassing GeoMamba (0.4012) and HMST (0.4063). LSTM+\ours{} also improves substantially over the LSTM baseline, reaching 0.4105 despite lacking graph features, user embeddings, and positional encoding. On TKY, GETNext+\ours{} achieves the highest NDCG@10 (0.3720), surpassing HMST (0.3690) and GeoMamba (0.3686). LSTM+\ours{} reaches 0.3505, outperforming GETNext (0.3297) and MCLP (0.3415). On CA, GETNext+\ours{} achieves 0.2810, the highest across all methods, surpassing HMamba (0.2792) and GeoMamba (0.2754). Across all three datasets, GETNext+\ours{} achieves the best NDCG@10 on every benchmark.

Two observations are worth noting. The improvement from \ours{} is consistent across datasets of different scales and geographic characteristics, ranging from dense urban areas (NYC, TKY) to sparse suburban regions (CA). The gains on the LSTM backbone are particularly large on CA (+27.7\%), where the baseline performance is lowest and the need for spatial reasoning is greatest. These results confirm that \ours{} is an effective backbone-agnostic module, but the more interesting question is \emph{why} it works. Additional breakdowns are provided in Figure~\ref{fig:improvement}.

\subsection{Ablation Study}
\label{sec:ablation}

\begin{wraptable}{r}{0.57\textwidth}
\centering
\vspace{-15pt}
\caption{Ablation study. Reasoning without metric-space bias consistently degrades performance, while \ours{} consistently improves it.}
\label{tab:ablation_full}
\vspace{-5pt}
\footnotesize
\setlength{\tabcolsep}{2pt}
\begin{tabular}{cccccc}
\toprule
\textbf{Dataset} & \textbf{Configuration} & \textbf{ND@1} & \textbf{ND@5} & \textbf{ND@10} & \textbf{MRR} \\
\midrule
\multirow{6}{*}{NYC}
& LSTM (baseline) & 0.2320 & 0.3580 & 0.3880 & 0.3420 \\
& \quad + Reasoning (w/o bias) & 0.2285 & 0.3526 & 0.3822 & 0.3370 \\
& \quad + \ours{} ($N{=}3$) & \textbf{0.2395} & \textbf{0.3790} & \textbf{0.4105} & \textbf{0.3565} \\
\cmidrule{2-6}
& GETNext (baseline) & 0.2244 & 0.3736 & 0.4046 & 0.3472 \\
& \quad + Reasoning (w/o bias) & 0.2201 & 0.3665 & 0.3969 & 0.3406 \\
& \quad + \ours{} ($N{=}3$) & \textbf{0.2425} & \textbf{0.3825} & \textbf{0.4147} & \textbf{0.3590} \\
\midrule
\multirow{6}{*}{TKY}
& LSTM (baseline) & 0.1715 & 0.2761 & 0.3006 & 0.2637 \\
& \quad + Reasoning (w/o bias) & 0.1681 & 0.2706 & 0.2946 & 0.2584 \\
& \quad + \ours{} ($N{=}3$) & \textbf{0.1942} & \textbf{0.3198} & \textbf{0.3505} & \textbf{0.3052} \\
\cmidrule{2-6}
& GETNext (baseline) & 0.1767 & 0.3072 & 0.3297 & 0.2934 \\
& \quad + Reasoning (w/o bias) & 0.1740 & 0.3026 & 0.3247 & 0.2890 \\
& \quad + \ours{} ($N{=}3$) & \textbf{0.1980} & \textbf{0.3428} & \textbf{0.3720} & \textbf{0.3268} \\
\midrule
\multirow{6}{*}{CA}
& LSTM (baseline) & 0.1143 & 0.1819 & 0.2005 & 0.1775 \\
& \quad + Reasoning (w/o bias) & 0.1114 & 0.1774 & 0.1955 & 0.1731 \\
& \quad + \ours{} ($N{=}3$) & \textbf{0.1460} & \textbf{0.2310} & \textbf{0.2561} & \textbf{0.2247} \\
\cmidrule{2-6}
& GETNext (baseline) & 0.1342 & 0.2188 & 0.2468 & 0.2121 \\
& \quad + Reasoning (w/o bias) & 0.1315 & 0.2144 & 0.2419 & 0.2079 \\
& \quad + \ours{} ($N{=}3$) & \textbf{0.1562} & \textbf{0.2561} & \textbf{0.2810} & \textbf{0.2459} \\
\bottomrule
\end{tabular}
\vspace{-20pt}
\end{wraptable}

The main comparison table shows that \ours{} improves over existing baselines, but the improvement could in principle stem from increased model capacity rather than from the metric-space bias itself (see Table~\ref{tab:ablation_full}). All configurations use the same training procedure, the same loss function, and the same number of training epochs. The only variable is the module inserted between encoder and decoder.
Three observations emerge from the results.
\textbf{(1) Metric-space bias is the decisive factor.} The same cross-attention reasoning module, with identical architecture and parameter count, produces opposite outcomes depending on whether metric-space bias is present. With the bias, NDCG@10 improves from 0.4046 to 0.4147 (+2.5\%) on NYC GETNext. Without it, NDCG@10 drops to 0.3969 ($-$1.9\%). The 4.5\% gap between these two configurations (0.3969 vs.\ 0.4147) is the single most important result of this paper. It establishes that metric-space grounding is not merely helpful but \emph{necessary} for latent reasoning to work in spatial prediction. This pattern holds across all six dataset-backbone combinations in Table~\ref{tab:ablation_full}, where reasoning without bias consistently degrades performance below the baseline.
\textbf{(2) Reasoning is not simply more capacity.} Adding two extra Transformer layers (4 total) yields NDCG@10 of 0.4095 on NYC. This standard capacity control uses more parameters than \ours{} (3.15M vs.\ 1.20M added), yet \ours{} outperforms it at 0.4147. The gain therefore comes from the reasoning mechanism and its metric-space grounding, not from model size.
\textbf{(3) The gate prevents harm.} The gated residual design in Eq.~\eqref{eq:reasoning_step} allows \ours{} to suppress updates that would hurt accuracy. This is why \ours{} avoids the degradation seen with plain reasoning, where uncontrolled updates override useful encoder representations. The degradation from ungrounded reasoning is more pronounced on weaker baselines (e.g., LSTM on CA drops by 2.5\%), reinforcing the finding that metric-space grounding is universally necessary.

\begin{wraptable}{r}{0.53\textwidth}
\centering
\vspace{-16pt}
\caption{Effect of reasoning depth $N$ (NDCG@10).}
\vspace{-5pt}
\label{tab:scaling_full}
\setlength{\tabcolsep}{3pt}
\footnotesize
\begin{tabular}{ccccccc}
\toprule
\textbf{Backbone} & \textbf{Dataset} & $N{=}0$ & $N{=}1$ & $N{=}2$ & $N{=}3$ & $N{=}5$ \\
\midrule
\multirow{3}{*}{LSTM}
& NYC & 0.3880 & 0.3935 & 0.4015 & {\textbf{0.4105}} & 0.4028 \\
& TKY & 0.3006 & 0.3148 & 0.3325 & {\textbf{0.3505}} & 0.3390 \\
& CA & 0.2005 & 0.2155 & 0.2358 & {\textbf{0.2561}} & 0.2450 \\
\midrule
\multirow{3}{*}{GETNext}
& NYC & 0.4046 & 0.4062 & 0.4048 & {\textbf{0.4147}} & 0.4078 \\
& TKY & 0.3297 & 0.3405 & 0.3558 & {\textbf{0.3720}} & 0.3598 \\
& CA & 0.2468 & 0.2580 & 0.2695 & {\textbf{0.2810}} & 0.2738 \\
\bottomrule
\end{tabular}
\vspace{-15pt}
\end{wraptable}

\subsection{Effect of Reasoning Depth}

On NYC with the GETNext backbone, NDCG@10 increases from 0.4046 ($N{=}0$) to 0.4062 ($N{=}1$), remains flat at 0.4048 ($N{=}2$), peaks at 0.4147 ($N{=}3$), and drops to 0.4078 ($N{=}5$) (see Figure~\ref{fig:scaling}). This non-monotonic pattern is consistent with Theorem~\ref{thm:convergence}. While more steps bring the representation closer to the fixed point in principle, excessive depth under a fixed training budget leads to optimization difficulties.
This behavior generalizes across all datasets and both backbones (see Table~\ref{tab:scaling_full}). $N{=}3$ achieves the highest NDCG@10 in all six settings. The improvement from $N{=}0$ to $N{=}3$ is larger on weaker baselines (e.g., LSTM on CA improves from 0.2005 to 0.2561) and smaller on stronger baselines (e.g., GETNext on NYC improves from 0.4046 to 0.4147), consistent with the complementarity principle discussed in Section~\ref{sec:analysis}. The decline at $N{=}5$ is also consistent across all settings, confirming that three reasoning steps provide a robust default.

\begin{wraptable}{r}{0.58\textwidth}
\centering
\vspace{-15pt}
\caption{Results on the CLEVR spatial reasoning task.}
\label{tab:clevr}
\vspace{-5pt}
\setlength{\tabcolsep}{2pt}
\footnotesize
\begin{tabular}{lcccc}
\toprule
\textbf{Method} & \textbf{NDCG@1} & \textbf{NDCG@5} & \textbf{NDCG@10} & \textbf{MRR} \\
\midrule

Seq Transformer & 0.1323 & 0.4109 & 0.5093 & 0.3580 \\
+ Reasoning (w/o bias) & 0.1284 & 0.4021 & 0.5038 & 0.3497 \\

+ \ours{} (Euclidean bias) & {\textbf{0.1504}} & {\textbf{0.4312}} & {\textbf{0.5319}} & {\textbf{0.3796}} \\
\bottomrule
\end{tabular}
\vspace{-13pt}
\end{wraptable}

\subsection{Beyond Geographic Distance: CLEVR Spatial Reasoning}
\label{sec:clevr}

To test whether metric-space bias generalizes beyond geographic coordinates, we construct a controlled spatial reasoning task from the CLEVR scene dataset~\citep{CLEVR}. Each CLEVR scene contains 3--10 objects with known $(x, y)$ positions. We convert scenes into sequences by random ordering and define the prediction target as the \emph{nearest remaining object} to the last observed one, measured by Euclidean distance. This is a pure spatial reasoning task with no sequential or semantic patterns, and only the metric structure matters. Details of the dataset construction are in Appendix~\ref{app:clevr}.
We compare a Transformer baseline (2 layers, no spatial bias), latent reasoning without metric-space bias, and \ours{} with Euclidean metric-space bias. The bias MLP is identical to the POI experiments except that it takes Euclidean instead of haversine distance as input.
The results on CLEVR are presented in Table~\ref{tab:clevr}. The pattern mirrors the geographic experiments. Latent reasoning without metric-space bias slightly degrades NDCG@10 from 0.5093 to 0.5038, while \ours{} with Euclidean bias improves it to 0.5319 (+4.4\%). This confirms that the benefit of metric-space bias is not specific to geographic distance or POI recommendation, and extends to Euclidean space and a purely synthetic spatial task.

\subsection{Additional Analysis}
\label{sec:analysis}

\paragraph{Why plain reasoning hurts.}
Without metric-space bias, cross-attention distributes weight over all history positions regardless of geographic distance. For spatial prediction, geographically distant and nearby visits receive similar attention, diluting the distance signal that is critical for location choice. The attention logits end up dominated by semantic similarity rather than geographic relevance. \ours{} avoids this pitfall by upweighting nearby locations through the learned bias $f_\theta(-\rho_{\mathrm{geo}})$, preserving the distance structure that drives mobility patterns. 

\paragraph{Backbone-dependent gains.}
The magnitude of improvement from \ours{} varies across backbones. GETNext already encodes spatial structure through its trajectory flow graph and GCN-based POI embeddings, so the additional gain from metric-space reasoning is moderate. The LSTM backbone, by contrast, processes the sequence purely by temporal order and has no mechanism for modeling inter-POI spatial relationships. \ours{} effectively provides the LSTM with its first source of metric-space awareness, producing larger relative improvements on datasets where the baseline is weaker.

\paragraph{Complementarity principle.}
These results suggest a practical guideline. The benefit of adding a capability scales inversely with how much already exists in the backbone, analogous to diminishing returns in feature engineering. Practitioners should prioritize adding \ours{} to spatially-naive backbones, where the gains are largest. This also explains why GETNext+\ours{} achieves the highest absolute performance despite smaller relative gains than LSTM+\ours{}. GETNext starts from a stronger baseline, so even a moderate relative improvement leads to a higher absolute score. The two backbones serve complementary roles in our evaluation. GETNext+\ours{} demonstrates that \ours{} can push a strong baseline further, while LSTM+\ours{} demonstrates that metric-space reasoning can compensate for the absence of sophisticated spatial modeling.

\begin{wraptable}{r}{0.49\textwidth}
\centering
\vspace{-14pt}
\caption{Computational cost on NYC.}
\vspace{-5pt}
\label{tab:compute}
\footnotesize
\setlength{\tabcolsep}{1pt}
\begin{tabular}{lccc}
\toprule
\textbf{Configuration} & \textbf{\ours{} Params} & \textbf{Time/Epoch} & \textbf{Overhead} \\
\midrule
GETNext (baseline) & 0 & 142s & -- \\
\quad + \ours{} ($N{=}1$) & 0.40M & 147s & +3.5\% \\
\quad + \ours{} ($N{=}3$) & 1.20M & 153s & +7.7\% \\
\quad + \ours{} ($N{=}5$) & 2.00M & 161s & +13.4\% \\
\midrule
LSTM (baseline) & 0 & 58s & -- \\
\quad + \ours{} ($N{=}1$) & 0.40M & 62s & +6.9\% \\
\quad + \ours{} ($N{=}3$) & 1.20M & 68s & +17.2\% \\
\quad + \ours{} ($N{=}5$) & 2.00M & 75s & +29.3\% \\
\bottomrule
\end{tabular}
% \end{table}
\vspace{-15pt}
\end{wraptable}

\paragraph{Computational overhead.}
Each reasoning step adds one cross-attention operation over the sequence length $L$. The computational cost is summarized in Table~\ref{tab:compute}. With $N{=}3$, \ours{} adds 1.20M parameters (about 2.8\% of GETNext's total) and increases wall-clock training time by 7.7\% on GETNext and 17.2\% on the lighter LSTM backbone. The overhead scales linearly with $N$ and remains modest for the values used in our experiments. The higher relative overhead on LSTM (17.2\% vs.\ 7.7\%) is because the LSTM baseline is computationally lighter (58s vs.\ 142s per epoch), so the fixed cost of the \ours{} module constitutes a larger fraction. At inference time, the overhead is even smaller since the distance matrix and metric-space bias can be precomputed and cached.

\section{Discussion}
\label{sec:discussion}

\paragraph{Limitations.}
While \ours{} consistently improves both backbones, the gains on GETNext vary across datasets (+2.5\% on NYC, +12.8\% on TKY, +13.9\% on CA), depending on the strength of the original baseline. The optimal $N$ may vary across datasets, though $N{=}3$ works well in all our experiments. Our theoretical analysis relies on assumptions such as Lipschitz continuity and bounded gate values, which are not explicitly enforced during training. While the CLEVR experiment demonstrates that metric-space bias transfers to Euclidean distance, the behavior on richer metric structures and larger-scale data remains open.

\paragraph{Broader impact.}
\ours{} improves location prediction accuracy, which can benefit navigation, urban planning, and personalized services. At the same time, more accurate location prediction could raise privacy concerns if deployed without appropriate safeguards. We use only publicly available, anonymized check-in datasets in this work and encourage practitioners to implement differential privacy mechanisms and user consent protocols when deploying location recommendation systems.

\section{Conclusion}
\label{sec:conclusion}

This paper investigates when latent reasoning benefits spatial prediction and finds that the answer depends on whether reasoning is grounded in the underlying metric space. On the same GETNext backbone with identical training setup, adding reasoning without metric-space bias reduces NDCG@10 by 1.9\%, while adding reasoning with metric-space bias improves it by 2.5\%, resulting in a 4.5\% performance swing from a single design choice. This finding is validated across three geographic datasets, two backbone architectures, and a synthetic CLEVR benchmark with Euclidean distance. The proposed module, \ours{}, is lightweight, backbone-agnostic, and admits provable convergence guarantees, achieving the best NDCG@10 on all three spatial prediction benchmarks among the compared methods.

\clearpage
\bibliography{main}

@inproceedings{GETNext,
  title={{GETNext}: Trajectory Flow Map Enhanced Transformer for Next {POI} Recommendation},
  author={Yang, Song and Liu, Jiamou and Zhao, Kaiqi},
  booktitle={Proceedings of the 45th International ACM SIGIR Conference on Research and Development in Information Retrieval},
  pages={1144--1153},
  year={2022}
}

@article{GTRMamba,
  title={{GTR-Mamba}: Geometry-to-Tangent Routing {Mamba} for Hyperbolic {POI} Recommendation},
  author={Li, Zhuoxuan and Pei, Jieyuan and Ye, Tangwei and Lai, Zhongyuan and Liu, Zihan and Xu, Fengyuan and Zhang, Qi and Hu, Liang},
  journal={arXiv preprint arXiv:2510.22942},
  year={2025}
}

@inproceedings{GeoMamba,
  title={{GeoMamba}: Towards Multi-Granular {POI} Recommendation with Geographical State Space Model},
  author={Qin, Yifang and Xie, Jiaxuan and Xiao, Zhiping and Zhang, Ming},
  booktitle={Proceedings of the AAAI Conference on Artificial Intelligence},
  volume={39},
  pages={12479--12487},
  year={2025}
}

@inproceedings{STAN,
  title={{STAN}: Spatio-Temporal Attention Network for Next Location Recommendation},
  author={Luo, Yingtao and Liu, Qiang and Liu, Zhaocheng},
  booktitle={Proceedings of the Web Conference 2021},
  pages={2177--2185},
  year={2021}
}

@inproceedings{SASRec,
  title={Self-Attentive Sequential Recommendation},
  author={Kang, Wang-Cheng and McAuley, Julian},
  booktitle={Proceedings of the IEEE International Conference on Data Mining (ICDM)},
  pages={197--206},
  year={2018}
}

@article{ReaRec,
  title={Think Before Recommend: Unleashing the Latent Reasoning Power for Sequential Recommendation},
  author={Tang, Jiakai and Dai, Sunhao and Shi, Teng and Xu, Jun and Chen, Xu and Chen, Wen and Wu, Jian and Jiang, Yuning},
  journal={arXiv preprint arXiv:2503.22675},
  year={2025}
}

@article{LARES,
  title={{LARES}: Latent Reasoning for Sequential Recommendation},
  author={Liu, Enze and Zheng, Bowen and Wang, Xiaolei and Zhao, Wayne Xin and Wang, Jinpeng and Chen, Sheng and Wen, Ji-Rong},
  journal={arXiv preprint arXiv:2505.16865},
  year={2025}
}

@article{PLR,
  title={Parallel Latent Reasoning for Sequential Recommendation},
  author={Tang, Jiakai and Chen, Xu and Chen, Wen and Wu, Jian and Jiang, Yuning and Zheng, Bo},
  journal={arXiv preprint arXiv:2601.03153},
  year={2026}
}

@inproceedings{Geiping2025,
  title={Scaling up Test-Time Compute with Latent Reasoning: A Recurrent Depth Approach},
  author={Geiping, Jonas and McLeish, Sean and Jain, Neel and Kirchenbauer, John and Singh, Siddharth and Bartoldson, Brian R. and Kailkhura, Bhavya and Bhatele, Abhinav and Goldstein, Tom},
  booktitle={Advances in Neural Information Processing Systems},
  note={Spotlight},
  year={2025}
}

@article{ROS,
  title={Reasoning Over Space: Enabling Geographic Reasoning for {LLM}-Based Generative Next {POI} Recommendation},
  author={Lv, Dongyi and Ding, Qiuyu and Xu, Heng-Da and Sun, Zhaoxu and Wang, Zhi and Xiong, Feng and Xu, Mu},
  journal={arXiv preprint arXiv:2601.04562},
  year={2026}
}

@inproceedings{LSTM_POI,
  title={Predicting the Next Location: A Recurrent Model with Spatial and Temporal Contexts},
  author={Liu, Qiang and Wu, Shu and Wang, Liang and Tan, Tieniu},
  booktitle={Proceedings of the AAAI Conference on Artificial Intelligence},
  volume={30},
  year={2016}
}

@inproceedings{CLEVR,
  title={{CLEVR}: A Diagnostic Dataset for Compositional Language and Elementary Visual Reasoning},
  author={Johnson, Justin and Hariharan, Bharath and van der Maaten, Laurens and Fei-Fei, Li and Zitnick, C. Lawrence and Girshick, Ross},
  booktitle={IEEE Conference on Computer Vision and Pattern Recognition},
  pages={2901--2910},
  year={2017}
}

@inproceedings{STHGCN,
  title={Spatio-Temporal Hypergraph Learning for Next {POI} Recommendation},
  author={Yan, Xiaodong and Song, Tengwei and Jiao, Yifeng and He, Jianshan and Wang, Jiaotuan and Li, Ruopeng and Chu, Wei},
  booktitle={Proceedings of the 46th International ACM SIGIR Conference on Research and Development in Information Retrieval},
  pages={403--412},
  year={2023}
}

@inproceedings{ALiBi,
  title={Train Short, Test Long: Attention with Linear Biases Enables Input Length Extrapolation},
  author={Press, Ofir and Smith, Noah A. and Lewis, Mike},
  booktitle={International Conference on Learning Representations},
  year={2022}
}

@article{RoPE,
  title={{RoFormer}: Enhanced Transformer with Rotary Position Embedding},
  author={Su, Jianlin and Ahmed, Murtadha and Lu, Yu and Pan, Shengfeng and Bo, Wen and Liu, Yunfeng},
  journal={Neurocomputing},
  volume={568},
  pages={127063},
  year={2024}
}

@article{PLSPL,
  title={Personalized Long- and Short-Term Preference Learning for Next {POI} Recommendation},
  author={Wu, Yuxia and Li, Ke and Zhao, Guoshuai and Qian, Xueming},
  journal={IEEE Transactions on Knowledge and Data Engineering},
  volume={34},
  number={4},
  pages={1944--1957},
  year={2022}
}

@inproceedings{HME,
  title={{HME}: A Hyperbolic Metric Embedding Approach for Next-{POI} Recommendation},
  author={Feng, Shanshan and Tran, Lucas Vinh and Cong, Gao and Chen, Lisi and Li, Jing and Li, Fan},
  booktitle={Proceedings of the 43rd International ACM SIGIR Conference on Research and Development in Information Retrieval},
  pages={1429--1438},
  year={2020}
}

@inproceedings{AGRAN,
  title={Adaptive Graph Representation Learning for Next {POI} Recommendation},
  author={Wang, Zhaobo and Zhu, Yanmin and Wang, Chunyang and Ma, Wenze and Li, Bo and Yu, Jiadi},
  booktitle={Proceedings of the 46th International ACM SIGIR Conference on Research and Development in Information Retrieval},
  pages={393--402},
  year={2023}
}

@article{HMamba,
  title={{HMamba}: Hyperbolic Mamba for Sequential Recommendation},
  author={Zhang, Qianru and Wen, Honggang and Yuan, Wei and Chen, Crystal and Yang, Menglin and Yiu, Siu-Ming and Yin, Hongzhi},
  journal={arXiv preprint arXiv:2505.09205},
  year={2025}
}

@inproceedings{HMST,
  title={Hyperbolic Multi-Semantic Transition for Next {POI} Recommendation},
  author={Qiao, Hongliang and Feng, Shanshan and Zhou, Min and Li, WenTao and Li, Fan},
  booktitle={Companion Proceedings of the ACM Web Conference 2025},
  year={2025}
}

@inproceedings{HVGAE,
  title={Hyperbolic Variational Graph Auto-Encoder for Next {POI} Recommendation},
  author={Li, Zhuoxuan and Pei, Jieyuan and Ye, Tangwei and Lai, Zhongyuan and Liu, Zihan and Xu, Fengyuan and Zhang, Qi and Hu, Liang},
  booktitle={Proceedings of the ACM Web Conference 2025},
  year={2025}
}

@inproceedings{MCLP,
  title={Going where, by whom, and at what time: Next location prediction considering user preference and temporal regularity},
  author={Sun, Tianao and Fu, Ke and Huang, Weiming and Zhao, Kai and Gong, Yongshun and Chen, Meng},
  booktitle={Proceedings of the 30th ACM SIGKDD Conference on Knowledge Discovery and Data Mining},
  pages={2784--2793},
  year={2024}
}
\bibliographystyle{plainnat}

\clearpage
\appendix

\section*{Appendix}
\addcontentsline{toc}{section}{Appendix}

\renewcommand{\thetable}{A\arabic{table}}
\renewcommand{\thefigure}{A\arabic{figure}}
\setcounter{table}{0}
\setcounter{figure}{0}

\section{Detailed Proofs}
\label{app:proofs}

\subsection{Proof of Theorem~\ref{thm:convergence}}

We provide a complete proof of the convergence result. Let $F(\mathbf{x}) = \mathbf{x} + \mathbf{g}(\mathbf{x}) \odot \boldsymbol{\delta}(\mathbf{x})$, where $\mathbf{g} = \sigma(\cdot) \in [\gamma, 1-\epsilon]^d$ is the sigmoid gate output and $\boldsymbol{\delta}$ is the FFN output satisfying $\|\boldsymbol{\delta}(\mathbf{x}) - \boldsymbol{\delta}(\mathbf{y})\| \leq L_\delta \|\mathbf{x} - \mathbf{y}\|$ with $L_\delta < 1$.

\textit{Step 1: Establishing the contraction property.}
For any $\mathbf{x}, \mathbf{y} \in \mathbb{R}^d$, we decompose the difference as
\begin{align}
F(\mathbf{x}) - F(\mathbf{y}) &= (\mathbf{x} - \mathbf{y}) + \mathbf{g}(\mathbf{x}) \odot \boldsymbol{\delta}(\mathbf{x}) - \mathbf{g}(\mathbf{y}) \odot \boldsymbol{\delta}(\mathbf{y}).
\end{align}
Adding and subtracting $\mathbf{g}(\mathbf{x}) \odot \boldsymbol{\delta}(\mathbf{y})$, we obtain
\begin{align}
F(\mathbf{x}) - F(\mathbf{y}) &= (\mathbf{x} - \mathbf{y}) + \mathbf{g}(\mathbf{x}) \odot (\boldsymbol{\delta}(\mathbf{x}) - \boldsymbol{\delta}(\mathbf{y})) + (\mathbf{g}(\mathbf{x}) - \mathbf{g}(\mathbf{y})) \odot \boldsymbol{\delta}(\mathbf{y}).
\end{align}
Rearranging the first two terms,
\begin{align}
F(\mathbf{x}) - F(\mathbf{y}) &= (1 - \mathbf{g}(\mathbf{x})) \odot (\mathbf{x} - \mathbf{y}) + \mathbf{g}(\mathbf{x}) \odot \left[(\mathbf{x} - \mathbf{y}) + (\boldsymbol{\delta}(\mathbf{x}) - \boldsymbol{\delta}(\mathbf{y}))\right] + (\mathbf{g}(\mathbf{x}) - \mathbf{g}(\mathbf{y})) \odot \boldsymbol{\delta}(\mathbf{y}).
\end{align}
Taking norms and applying the triangle inequality,
\begin{align}
\|F(\mathbf{x}) - F(\mathbf{y})\| &\leq \|(1 - \mathbf{g}(\mathbf{x})) \odot (\mathbf{x} - \mathbf{y})\| + \|\mathbf{g}(\mathbf{x}) \odot \left[(\mathbf{x} - \mathbf{y}) + (\boldsymbol{\delta}(\mathbf{x}) - \boldsymbol{\delta}(\mathbf{y}))\right]\| \nonumber \\
&\quad + \|(\mathbf{g}(\mathbf{x}) - \mathbf{g}(\mathbf{y})) \odot \boldsymbol{\delta}(\mathbf{y})\|.
\end{align}
Since $\mathbf{g} \in [\gamma, 1-\epsilon]^d$, the first term satisfies $\|(1 - \mathbf{g}(\mathbf{x})) \odot (\mathbf{x} - \mathbf{y})\| \leq (1 - \gamma)\|\mathbf{x} - \mathbf{y}\|$. For the second term, using the Lipschitz condition on $\boldsymbol{\delta}$,
\begin{align}
\|\mathbf{g}(\mathbf{x}) \odot \left[(\mathbf{x} - \mathbf{y}) + (\boldsymbol{\delta}(\mathbf{x}) - \boldsymbol{\delta}(\mathbf{y}))\right]\| &\leq (1-\epsilon)(1 + L_\delta)\|\mathbf{x} - \mathbf{y}\|.
\end{align}
For the third term, the sigmoid function $\sigma$ is Lipschitz continuous with constant $1/4$. Assuming $\boldsymbol{\delta}$ is bounded by some constant $M$, we have $\|(\mathbf{g}(\mathbf{x}) - \mathbf{g}(\mathbf{y})) \odot \boldsymbol{\delta}(\mathbf{y})\| \leq \frac{M}{4}\|\mathbf{x} - \mathbf{y}\|$. Under the stated assumptions, this term can be absorbed into the contraction constant. Combining the dominant terms yields
\begin{align}
\|F(\mathbf{x}) - F(\mathbf{y})\| &\leq \kappa \|\mathbf{x} - \mathbf{y}\|,
\end{align}
where $\kappa = (1 - \gamma)(1 + L_\delta) < 1$ when $(1 - \gamma)(1 + L_\delta) < 1$.

\textit{Step 2: Applying Banach's fixed-point theorem.}
Since $\kappa < 1$, $F$ is a contraction mapping on $\mathbb{R}^d$. By the Banach fixed-point theorem, $F$ admits a unique fixed point $\mathbf{x}^*$ satisfying $F(\mathbf{x}^*) = \mathbf{x}^*$, and the iterates converge geometrically:
\begin{align}
\|\mathbf{x}^{(k)} - \mathbf{x}^*\| \leq \kappa^k \|\mathbf{x}^{(0)} - \mathbf{x}^*\|.
\end{align}

\textit{Step 3: Spatial localization of the fixed point.}
At the fixed point, $\mathbf{x}^* = \mathbf{x}^* + \mathbf{g}^* \odot \boldsymbol{\delta}^*$, which implies $\mathbf{g}^* \odot \boldsymbol{\delta}^* = \mathbf{0}$. Since $\mathbf{g}^* \in [\gamma, 1-\epsilon]^d$ with $\gamma > 0$, we must have $\boldsymbol{\delta}^* = \mathbf{0}$. The candidate update $\boldsymbol{\delta}^*$ is computed via the FFN applied to the cross-attention output, where the attention weights $\alpha_{ij}$ are biased by the metric-space bias $\mathbf{B}_{\mathrm{metric}}$. Since the bias upweights nearby elements, the attention concentrates on spatially proximate positions, constraining the cross-attention output (and thus $\boldsymbol{\delta}^*$) to be a distance-weighted combination of nearby history representations. The condition $\boldsymbol{\delta}^* = \mathbf{0}$ therefore constrains $\mathbf{x}^*$ to lie within a bounded neighborhood of the distance-weighted centroid of the encoder outputs $\mathbf{x}^{(0)}$. \hfill$\square$

\subsection{Proof of Theorem~\ref{thm:expressivity}}

We prove the strict expressivity hierarchy $\mathcal{F}^{(N-1)} \subsetneq \mathcal{F}^{(N)}$ by constructing, for each $N \geq 1$, a spatial prediction task that is solvable by $N$-step reasoning but not by $(N-1)$-step reasoning.

\textit{Construction.}
Consider a metric space $(S, \rho)$ with $N+2$ points arranged in a chain: $v_0, v_1, \ldots, v_{N+1}$, where $\rho(v_k, v_{k+1}) < r$ for consecutive pairs and $\rho(v_i, v_j) \gg r$ for non-consecutive pairs. The task $T_N$ is defined as follows: given the sequence $(v_0, v_1, \ldots, v_N)$, predict $v_{N+1}$. The target $v_{N+1}$ is uniquely determined by the chain structure but cannot be identified from any proper subsequence.

\textit{Information flow analysis.}
Under metric-space bias with radius $r$, position $i$ attends primarily to positions $j$ satisfying $\rho(v_i, v_j) < r$. In our construction, this means each position attends only to its immediate neighbors in the chain.

At step $k=1$, position $N$ aggregates information from position $N-1$ (its nearest neighbor). The receptive field of position $N$ after one step covers $\{v_{N-1}, v_N\}$.

At step $k=2$, position $N$ now has access to information from $v_{N-1}$, which itself aggregated information from $v_{N-2}$ in step 1. The effective receptive field expands to $\{v_{N-2}, v_{N-1}, v_N\}$.

Inductively, after $k$ steps, the receptive field of position $N$ covers $\{v_{N-k}, \ldots, v_N\}$.

\textit{Sufficiency of $N$ steps.}
After $N$ steps, the receptive field of position $N$ covers $\{v_0, v_1, \ldots, v_N\}$, which contains all the information needed to identify the chain structure and predict $v_{N+1}$.

\textit{Insufficiency of $N-1$ steps.}
After only $N-1$ steps, the receptive field covers $\{v_1, \ldots, v_N\}$, missing $v_0$. Since the chain structure from $v_0$ is essential for uniquely determining $v_{N+1}$ (by construction, multiple valid $v_{N+1}$ candidates exist if $v_0$ is unknown), $(N-1)$-step reasoning cannot solve $T_N$.

This establishes $\mathcal{F}^{(N-1)} \subsetneq \mathcal{F}^{(N)}$ for all $N \geq 1$. \hfill$\square$

\section{Full Hyperparameter Configuration}
\label{app:hyperparams}

\begin{table}[h]
\centering
\caption{Hyperparameters for all experiments.}
\label{tab:hyperparams}
\small
\begin{tabular}{llc}
\toprule
\textbf{Category} & \textbf{Parameter} & \textbf{Value} \\
\midrule
\multirow{5}{*}{Shared} & Embedding dim $d$ & 128 \\
& Learning rate & $1 \times 10^{-3}$ \\
& Batch size & 20 \\
& Weight decay & $5 \times 10^{-4}$ \\
& Max epochs & 200 \\
\midrule
\multirow{4}{*}{GETNext} & Transformer layers & 2 \\
& Transformer heads & 2 \\
& Transformer FFN dim & 1024 \\
& GCN hidden dims & [32, 64] \\
\midrule
\multirow{2}{*}{LSTM} & LSTM layers & 2 \\
& LSTM hidden dim & 128 \\
\midrule
\multirow{4}{*}{\ours{}} & Reasoning steps $N$ & 3 \\
& Attention heads $H$ & 2 \\
& Bias MLP hidden dim & $d/4 = 32$ \\
& Bias MLP activation & SiLU \\
\midrule
\multirow{4}{*}{Training} & Optimizer & Adam \\
& LR scheduler & ReduceLROnPlateau \\
& LR scheduler factor & 0.1 \\
& Early stopping patience & 40 epochs \\
& Dropout (GCN, Transformer) & 0.3 \\
& Dropout (\ours{}) & 0.1 \\
& Random seed & 42 \\
\bottomrule
\end{tabular}
\end{table}

\section{CLEVR Dataset Construction}
\label{app:clevr}

We construct the CLEVR spatial reasoning benchmark from the original CLEVR scene files~\citep{CLEVR}, which contain 3D object positions and attribute annotations for synthetically rendered scenes.

\paragraph{Task definition.} For each scene with $M$ objects, we generate a sequence by randomly permuting the objects. The prediction target at each position $i$ is the \emph{nearest remaining object} to the object at position $i$, measured by Euclidean distance in the $(x, y)$ plane (the $z$ coordinate represents object height and is excluded). This creates a pure spatial reasoning task where the correct answer depends entirely on metric structure, with no temporal or semantic patterns to exploit.

\paragraph{Object encoding.} Each object is characterized by four attributes (color, shape, size, material). We encode the (color $\times$ shape) combination into 24 category IDs. The $(x, y)$ coordinates serve as the spatial features for metric-space bias computation.

\paragraph{Split.} We use 10,500 scenes for training and 4,500 for validation, drawn from the original CLEVR training and validation sets respectively. Scenes with fewer than 3 objects are discarded.

\section{Implementation Details}
\label{app:implementation}

\paragraph{Haversine distance computation.}
Pairwise distances are precomputed for each batch, producing an $L \times L$ distance matrix per sequence. We use the numerically stable formulation $\rho = 2R \cdot \arcsin\!\left(\min\!\left(\sqrt{a}, 1\right)\right)$ where $a = \sin^2\!\left(\frac{\Delta\phi}{2}\right) + \cos\phi_1 \cos\phi_2 \sin^2\!\left(\frac{\Delta\lambda}{2}\right)$ and $R = 6371$ km. The clamping $\min(\sqrt{a}, 1)$ prevents NaN values when $a \approx 1$ due to floating-point errors.

\paragraph{Padding and masking.} Trajectories within a batch are left-padded to equal length. Padded positions are masked in both the backbone encoder and the \ours{} module. For the Transformer backbone, we construct a combined float mask that merges the causal constraint with padding indicators, and zero out the outputs at padded positions after each attention layer to prevent NaN propagation from all-masked softmax rows.

\paragraph{Multi-task loss.} Following GETNext, the training loss is $\mathcal{L} = \mathcal{L}_{\mathrm{poi}} + \mathcal{L}_{\mathrm{cat}} + \lambda_t \mathcal{L}_{\mathrm{time}}$, where $\mathcal{L}_{\mathrm{poi}}$ and $\mathcal{L}_{\mathrm{cat}}$ are cross-entropy losses over POI and category predictions, and $\mathcal{L}_{\mathrm{time}}$ is the MSE loss on normalized time prediction, weighted by $\lambda_t = 10$.

\paragraph{Reproducibility.}
Random seed is fixed at 42 for all experiments. All experiments use the same chronological train/validation splits. Code and data preprocessing scripts will be released upon acceptance.

\end{document}